\documentclass[11pt]{article}
\usepackage{geometry}
\usepackage{amsmath,amsfonts,amssymb,amsthm}
\usepackage{graphicx}
\usepackage{enumerate}
\usepackage{bbm}
\usepackage{tikz}
\usepackage{verbatim}
\usepackage{hyperref,color}
\hypersetup{colorlinks=true,citecolor=blue, linkcolor=blue, urlcolor=blue}

\newtheorem{thm}{Theorem}

\newtheorem{lemma}[thm]{Lemma}

\newtheorem{claim}[thm]{Claim}

\theoremstyle{definition}

\newcommand{\dif}{\mathrm{d}}

\newcommand{\Exp}{\mathbb{E}}

\newcommand{\norm}[1]{\left\lVert #1 \right\rVert}

\newcommand{\ip}[2]{\left\langle #1 , #2 \right\rangle}
\newcommand{\grad}{\nabla}

\newcommand{\Trel}{\tau_{\mathrm{rel}}}

\newcommand{\eps}{\varepsilon}

\newcommand{\R}{\mathbb{R}}

\begin{document}

\title{\textbf{Optimal Convergence Rate of Hamiltonian Monte Carlo for Strongly Logconcave Distributions}}

\author{Zongchen Chen$^*$ \and Santosh S. Vempala\thanks{Georgia Tech. Email: {\tt \{chenzongchen, vempala\}@gatech.edu}}}
\date{\today}

\maketitle

\begin{abstract}
	We study \emph{Hamiltonian Monte Carlo} (HMC) for sampling from a strongly logconcave density proportional to $e^{-f}$ where 
$f:\R^d \to \R$ is $\mu$-strongly convex and $L$-smooth (the condition number is $\kappa = L/\mu$). We show that the relaxation time (inverse of the spectral gap) of ideal HMC is $O(\kappa)$, improving on the previous best bound of $O(\kappa^{1.5})$; we complement this with an example where the relaxation time is $\Omega(\kappa)$. When implemented using a nearly optimal ODE solver, HMC returns an $\eps$-approximate point in 2-Wasserstein distance using $\widetilde{O}((\kappa d)^{0.5} \eps^{-1})$ gradient evaluations per step and $\widetilde{O}((\kappa d)^{1.5}\eps^{-1})$ total time.
\end{abstract}


\section{Introduction}

Sampling logconcave densities is a basic problem that arises in machine learning, statistics, optimization, computer science and other areas. The problem is described as follows. Let $f: \R^d \to \R$ be a convex function. Our goal is to sample from the density proportional to $e^{-f(x)}$. We study \textit{Hamiltonian Monte Carlo} (HMC), one of the most widely-used \textit{Markov chain Monte Carlo} (MCMC) algorithms for sampling from a probability distribution. In many settings, HMC is believed to outperform other MCMC algorithms such as the Metropolis-Hastings algorithm or Langevin dynamics. 
In terms of theory, rapid mixing has been established for HMC in recent papers \cite{LSV18,LV18,MS17,MS19,MV18} under various settings. However, in spite of much progress, there is a gap between known upper and lower bounds even in the basic setting when $f$ is strongly convex ($e^{-f}$ is strongly logconcave) and has a Lipschitz gradient.  


Many sampling algorithms such as the Metropolis-Hastings algorithm or Langevin dynamics maintain a position $x=x(t)$ that changes with time, so that the distribution of $x$ will eventually converge to the desired distribution, i.e., proportional to $e^{-f(x)}$.  
In HMC, besides the position $x=x(t)$, we also maintain a velocity $v=v(t)$. In the simplest Euclidean setting, the Hamiltonian $H(x,v)$ is defined as
\[
	H(x,v) = f(x) + \frac{1}{2} \|v\|^2.
\]
Then in every step the pair $(x,v)$ is updated using the following system of differential equations for a fixed time interval $T$:
\begin{equation}\label{eq:ODE_system}
	\left\{
	\begin{aligned}
	\frac{\dif x(t)}{\dif t} &= \frac{\partial H(x,v)}{\partial v} = v(t),\\
	\frac{\dif v(t)}{\dif t} &= -\frac{\partial H(x,v)}{\partial x} = -\nabla f(x(t)).
	\end{aligned}
	\right.
\end{equation}
The initial position $x(0) = x_0$ is the position from the last step, and the initial velocity $v(0) = v_0$ is chosen randomly from the standard Gaussian distribution $N(0,I)$. The updated position is $x(T)$ where $T$ can be thought of as the step-size. 
It is well-known that the stationary distribution of HMC is the density proportional to $e^{-f}$. 
Observe that 
\[
	\frac{\dif H(x,v)}{\dif t} = \frac{\partial H(x,v)}{\partial x} x'(t) + \frac{\partial H(x,v)}{\partial v} v'(t) = 0, 
\] 
so the Hamiltonian $H(x,v)$ does not change with $t$. 
We can also write \eqref{eq:ODE_system} as the following ordinary differential equation (ODE):
\begin{equation}\label{eq:ODE_Ham}
	x''(t) = -\nabla f(x(t)), \quad x(0) = x_0, \quad x'(0) = v_0.
\end{equation}
We state HMC explicitly as the following algorithm.

\bigskip

\noindent
\textbf{Hamiltonian Monte Carlo algorithm}

\medskip
Input: $f:\R^d \to \R$ that is $\mu$-strongly convex and $L$-smooth, $\eps$ the error parameter.
\begin{enumerate}
	\item Set starting point $x^{(0)}$, step-size $T$, number of steps $N$, and ODE error tolerance $\delta$.
	\item For $k=1,\dots,N$:
	\begin{enumerate}
		\item Let $v \sim N(0,I)$;
		\item Denote by $x(t)$ the solution to \eqref{eq:ODE_system} with initial position $x(0) = x^{(k-1)}$ and initial velocity $v(0) = v$. Use the ODE solver to find a point $x^{(k)}$ such that
		\[
			\norm{x^{(k)} - x(T)} \le \delta.
		\] 
	\end{enumerate}
	\item Output $x^{(N)}$. 
\end{enumerate}

In our analysis, we first consider {\em ideal} HMC where in every step we have the exact solution to the ODE \eqref{eq:ODE_system} and neglect the numerical error from solving the ODEs or integration ($\delta = 0$). 

\subsection{Preliminaries}
We recall standard definitions here. Let $f:\R^d \to \R$ be a continuously differentiable function. 
We say $f$ is $\mu$-strongly convex if for all $x,y\in \R^d$, 
$$ f(y) \ge f(x) + \ip{\grad f(x)}{y-x} + \frac{\mu}{2} \norm{y-x}^2. $$
We say $f$ is $L$-smooth if $\nabla f$ is $L$-Lipschitz; i.e., for all $x,y \in \R^d$,
$$ \norm{\grad f(x) - \grad f(y)} \le L \norm{x-y}. $$ 
If $f$ is $\mu$-strongly convex and $L$-smooth, then the condition number of $f$ is $\kappa = L/\mu$.

Consider a discrete-time \emph{reversible} Markov chain $\mathcal{M}$ on $\R^d$ with stationary distribution $\pi$. 
Let
\[
	L_2(\pi) = \left\{ f: \R^d \to \R \,\middle|\, \int_{\R^d} f(x)^2 \pi(\dif x) < \infty \right\}
\]
be the Hilbert space with inner product
\[
	\ip{f}{g} = \int_{\R^d} f(x) g(x) \pi(\dif x)
\]
for $f,g \in L_2(\pi)$. 
Denote by $P$ the transition kernel of $\mathcal{M}$. 
We can view $P$ as a self-adjoint operator from $L_2(\pi)$ to itself: for $f\in L_2(\pi)$, 
\[
(Pf) (x) = \int_{\R^d} f(y)P(x,\dif y).
\]
Let $L_2^0(\pi) = \{ f\in L_2(\pi): \int_{\R^d} f(x) \pi(\dif x) = 0 \}$ be a closed subspace of $L_2(\pi)$. The (absolute) \emph{spectral gap} of $P$ is defined to be
\[
	\gamma(P) = 1 - \sup_{f\in L_2^0(\pi)}\frac{\norm{Pf}}{\norm{f}} = 1 - \sup_{\substack{f\in L_2^0(\pi)\\ \norm{f} = 1}} |\ip{Pf}{f}|.
\]
The relaxation time of $P$ is 
\[
	\Trel(P) = \frac{1}{\gamma(P)}.
\]

Let $\nu_1,\nu_2$ be two distributions on $\R^d$. The $2$-Wasserstein distance between $\nu_1$ and $\nu_2$ is defined as
\[
	W_2(\nu_1,\nu_2) = \left( \inf_{(X,Y) \in \mathcal{C}(\nu_1,\nu_2)} \Exp\left[ \norm{X-Y}^2 \right] \right)^{1/2},
\]
where $\mathcal{C}(\nu_1,\nu_2)$ is the set of all couplings of $\nu_1$ and $\nu_2$. 

\subsection{Related work}
Various versions of Langevin dynamics have been studied in many recent papers, see \cite{D17,DK19, ZLC17,RRT17, DRD18,CCBJ, CCABJ18,CFMBJ18, DCWY18,W18, VW19,MCCFBJ19}. 
The convergence rate of HMC is also studied recently in \cite{LSV18,LV18,MS17, MS19,MV18,SRH14}. 
The first bound for our setting was obtained by Mangoubi and Smith \cite{MS17}, who gave an $O(\kappa^2)$ bound on the convergence rate of ideal HMC. 

\begin{thm}[{\cite[Theorem 1]{MS17}}]
	Let $f:\R^d \to\R$ be a twice differentiable function such that $\mu I \preceq \nabla^2 f(x) \preceq L I$ for all $x\in\R^d$. 
	Then the relaxation time of ideal HMC for sampling from the density $\propto e^{-f}$ with step-size $T = \sqrt{\mu}/(2\sqrt{2} L)$ is $O(\kappa^{2})$.  
\end{thm}

This was improved by \cite{LSV18}, which showed a bound of $O(\kappa^{1.5})$. They also gave a nearly optimal method for solving the ODE that arises in the implementation of HMC.  

\begin{thm}[{\cite[Lemma 1.8]{LSV18}}]
	Let $f:\R^d \to\R$ be a twice differentiable function such that $\mu I \preceq \nabla^2 f(x) \preceq L I$ for all $x\in\R^d$. 
	Then the relaxation time of ideal HMC for sampling from the density $\propto e^{-f}$ with step-size $T = \mu^{1/4}/(2L^{3/4})$ is $O(\kappa^{1.5})$.  
\end{thm}

Both papers suggest that the correct bound is linear in $\kappa$: \cite{MS17} says linear is the best one can expect while \cite{LSV18} shows that there {\em exists} a choice of step-sizes (time for running the ODE) that might achieve a linear rate (Lemma 1.8, second part); however it was far from clear how to determine these step-sizes algorithmically.  


Other papers focus on various aspects and use stronger assumptions (e.g., bounds on higher-order gradients) to get better bounds on the overall convergence time or the number of gradient evaluations in some ranges of parameters. For example, \cite{MV18} shows that the dependence on dimension for the number of gradient evaluations can be as low as $d^{1/4}$ with suitable regularity assumptions (and higher dependence on the condition number). We note also that sampling logconcave functions is a polynomial-time solvable problem, without the assumptions of strong convexity or gradient Lipschitzness, and even when the function $e^{-f}$ is given only by an oracle with no access to gradients \cite{AK91, LV06}. The Riemannian version of HMC provides a faster polynomial-time algorithm for uniformly sampling polytopes \cite{LV18}. However, the dependence on the dimension is significantly higher for these algorithms, both for the contraction rate and the time per step.

\subsection{Results}
In this paper, we show that the relaxation time of ideal HMC is $\Theta(\kappa)$ for strongly logconcave functions with Lipschitz gradient.  
\begin{thm}\label{thm:rel_upper}
	Suppose that $f$ is $\mu$-strongly convex and $L$-smooth. Then the relaxation time (inverse of spectral gap) of ideal HMC for sampling from the density $\propto e^{-f}$ with step-size $T = 1/(2\sqrt{L})$ is $O(\kappa)$, 
	where $\kappa = L/\mu$ is the condition number. 
\end{thm}
\noindent

We remark that the only assumption we made about $f$ is strongly convexity and smoothness (in particular, we do not require that $f$ is twice differentiable, which is assumed in both \cite{LSV18} and \cite{MS17}). 

We also establish a matching lower bound on the relaxation time of ideal HMC, implying the tightness of Theorem~\ref{thm:rel_upper}.
\begin{thm}\label{thm:rel_lower}
	For any $0 < \mu \le L$, there exists a $\mu$-strongly convex and $L$-smooth function $f$, such that the relaxation time of ideal HMC for sampling from the density $\propto e^{-f}$ with step-size $T = O(1/\sqrt{L})$ is $\Omega(\kappa)$, where $\kappa = L/\mu$ is the condition number. 
\end{thm}

Using the nearly optimal ODE solver from \cite{LSV18}, we obtain the following convergence rate in 2-Wasserstein distance for the HMC algorithm. We note that since our new convergence rate allows larger steps, the ODE solver is run for a longer time step. 
\begin{thm}
	\label{thm:mixing}
	Let $f:\R^d \to\R$ be a twice differentiable function such that $\mu I \preceq \nabla^2 f(x) \preceq L I$ for all $x\in\R^d$. 
	Let $\pi \propto e^{-f}$ be the target distribution, and let $\pi_{\textsc{hmc}}$ be the distribution of the output of HMC with starting point $x^{(0)} = \arg\min_x f(x)$, step-size $T = 1/(16000\sqrt{L})$, and 
	ODE error tolerance $\delta = \sqrt{\mu}T^2 \eps/16$.
	For any $0<\eps<\sqrt{d}$, if we run HMC for $N = O\left( \kappa \log(d/\eps) \right)$ steps where $\kappa = L/\mu$, then we have 
	\[
		W_2(\pi_{\textsc{hmc}},\pi) \le \frac{\eps}{\sqrt{\mu}}.
	\]
	Each step takes $O\left( \sqrt{\kappa}d^{3/2}\eps^{-1} \log(\kappa d/\eps) \right)$ time and $O\big( \sqrt{\kappa d}\eps^{-1} \log(\kappa d/\eps) \big)$ evaluations of $\grad f$, amortized over all steps.
\end{thm}
\noindent
The comparison of convergence rates, running times and numbers of gradient evaluations is summarized in the following table with polylog factors omitted. 

\begin{center}
{\renewcommand{\arraystretch}{1.3}%
\begin{tabular}{c|c|c|c}
\hline
reference & convergence rate & \# gradients & total time\\
\hline
\cite{MS17} & $\kappa^2$ & $\kappa^{6.5} d^{\,0.5}$  & $\kappa^{6.5} d^{1.5}$ \\
\cite{LSV18} & $\kappa^{1.5}$ & $\kappa^{1.75} d^{\,0.5}$ &  $\kappa^{1.75} d^{1.5}$ \\
this paper & $\kappa$ & $\kappa^{1.5} d^{\,0.5}$ & $\kappa^{1.5} d^{1.5}$\\
\hline
\end{tabular}
}
\end{center}

\section{Convergence of ideal HMC}
In this section we show that the spectral gap of ideal HMC is $\Omega(1/\kappa)$, and thus prove Theorem~\ref{thm:rel_upper}. 
We first show a contraction bound for ideal HMC, which roughly says that the distance of two points is shrinking after one step of ideal HMC. 

\begin{lemma}[Contraction bound]
	\label{lem:HMC_contraction}
	Suppose that $f$ is $\mu$-strongly convex and $L$-smooth. Let $x(t)$ and $y(t)$ be the solution to \eqref{eq:ODE_system} with initial positions $x(0)$, $y(0)$ and initial velocities $x'(0) = y'(0)$. Then for $0\le t \le 1/(2\sqrt{L})$ we have
	\[
	\norm{x(t) - y(t)}^2 \le \left( 1 - \frac{\mu}{4} t^2 \right) \norm{x(0) - y(0)}^2.
	\]
	In particular, by setting $t = T = 1/(c\sqrt{L})$ for some constant $c \ge 2$ we get
	\[
	\norm{x(T) - y(T)}^2 \le \left( 1 - \frac{1}{4c^2 \kappa} \right) \norm{x(0) - y(0)}^2
	\]
	where $\kappa = L/\mu$.
\end{lemma}

\begin{proof}
	Consider the two ODEs for HMC:
	\[
		\left\{\,
		\begin{aligned}
			x'(t) &= u(t);\\
			u'(t) &= -\grad f(x(t)).
		\end{aligned}
		\right.
		\qquad\text{and}\qquad
		\left\{\,
		\begin{aligned}
			y'(t) &= v(t);\\
			v'(t) &= -\grad f(y(t)).
		\end{aligned}
		\right.
	\]
	with initial points $x(0), y(0)$ and initial velocities $u(0) = v(0)$. 
	For the sake of brevity, we shall write $x = x(t)$, $y = y(t)$, $u = u(t)$, $v = v(t)$ and omit the variable $t$, as well as letting $x_0 = x(0)$, $y_0 = y(0)$. 
	We are going to show that
	\[
		\norm{x-y}^2 \leq \left( 1 - \frac{\mu}{4} t^2 \right) \norm{x_0 - y_0}^2
	\]
	for all $0 \le t \le 1/(2\sqrt{L})$. 
	
	Consider the derivative of $\frac{1}{2} \norm{x-y}^2$:
	\begin{equation}\label{eqn:deri_x-y}
		\frac{\dif}{\dif t} \left( \frac{1}{2} \norm{x-y}^2 \right) = 
		\ip{x'-y'}{x-y} = \ip{u-v}{x-y}.
	\end{equation}
	Taking derivative on both sides, we get
	\begin{align}
		\frac{\dif^2}{\dif t^2} \left( \frac{1}{2} \norm{x-y}^2 \right) 
		&= \ip{u'-v'}{x-y} + \ip{u-v}{x'-y'} \nonumber\\
		&= - \ip{\grad f(x) - \grad f(y)}{x-y} + \norm{u-v}^2 \nonumber\\ 
		&= - \rho \norm{x-y}^2 + \norm{u-v}^2, \label{eqn:deri2_x-y}
	\end{align}
	where we define
	\[
		\rho = \rho(t) = \frac{\ip{\grad f(x) - \grad f(y)}{x-y}}{\norm{x-y}^2}.
	\]
	Since $f$ is $\mu$-strongly convex and $L$-smooth, we have $\mu \le \rho \le L$ for all $t\geq 0$.
	
	We will upper bound the term $- \rho\norm{x-y}^2 + \norm{u-v}^2$, while keeping its dependency on $\rho$. To lower bound $\norm{x-y}^2$, we use the following crude bound.
	\begin{claim}[Crude bound]
		\label{claim:crude-bound}
		For all $0 \le t \le 1/(2\sqrt{L})$ we have
		\begin{equation}\label{eqn:crude_bound}
			\frac{1}{2} \norm{x_0 - y_0}^2 \le \norm{x-y}^2 \le 2 \norm{x_0 - y_0}^2.
		\end{equation}
	\end{claim}
	\noindent
	The proof of this claim is postponed to Section~\ref{subsec:proof_claim}. 

	Next we derive an upper bound on $\norm{u-v}^2$. 
	The derivative of $\norm{u-v}$ is given by
	\[
		\norm{u-v} \left( \frac{\dif}{\dif t} \norm{u-v} \right) = \frac{\dif^2}{\dif t^2} \left( \frac{1}{2} \norm{u-v}^2 \right) = \ip{u'-v'}{u-v} = - \ip{\grad f(x) - \grad f(y)}{u-v}.
	\]
	Thus, its absolute value is upper bounded by
	\[
		\left| \frac{\dif}{\dif t} \norm{u-v} \right| = \frac{\left| - \ip{\grad f(x) - \grad f(y)}{u-v} \right|}{\norm{u-v}} \le \norm{\grad f(x) - \grad f(y)}.
	\]
	Since $f$ is $L$-smooth and convex, we have
	\[
		\norm{\grad f(x) - \grad f(y)}^2 \le L \ip{\grad f(x) - \grad f(y)}{x-y} = L\rho \norm{x-y}^2 \le 2L\rho \norm{x_0-y_0}^2,
	\]
	where the last inequality follows from the crude bound \eqref{eqn:crude_bound}. 
	Then, using the fact that $u_0=v_0$ and the Cauchy-Schwarz inequality, we can upper bound $\norm{u-v}^2$ by
	\begin{align*}
		\norm{u-v}^2 &\le \left( \int_{0}^t \left| \frac{\dif}{\dif s} \norm{u-v} \right| \dif s \right)^2\\ 
		&\le \left( \int_{0}^t \sqrt{2L\rho} \norm{x_0-y_0} \dif s \right)^2\\ 
		&\le 2L t \left( \int_{0}^t \rho \,\dif s \right) \norm{x_0-y_0}^2. 
	\end{align*}
	Define the function
	\[
		P = P(t) = \int_{0}^t \rho \,\dif s,
	\]
	so $P(t)$ is nonnegative and monotone increasing, with $P(0) = 0$. Also we have $\mu t \le P(t) \le Lt$ for all $t\geq 0$. Then, 
	\begin{equation}\label{eqn:bound_u-v}
		\norm{u-v}^2 \le 2Lt P \norm{x_0-y_0}^2.
	\end{equation}
	
	Plugging \eqref{eqn:crude_bound} and \eqref{eqn:bound_u-v} into \eqref{eqn:deri2_x-y}, we deduce that
	\begin{equation*}
		\frac{\dif^2}{\dif t^2} \left( \frac{1}{2} \norm{x-y}^2 \right) 
		\le - \rho \left( \frac{1}{2} \norm{x_0 - y_0}^2 \right) + 2Lt P \norm{x_0-y_0}^2. 
	\end{equation*}
	If we define
	\[
		\alpha(t) = \frac{1}{2} \norm{x-y}^2, 
	\]
	then we have 
	\[
		\alpha''(t) \le - \alpha(0) \big( \rho(t)  - 4L t P(t) \big).
	\]
	Integrating both sides and using $\alpha'(0) = 0$, we obtain
	\begin{align*}
		\alpha'(t) &= \int_{0}^t \alpha''(s) \dif s\\ 
		&\le - \alpha(0) \left( \int_{0}^t \rho(s) \dif s - 4L \int_{0}^t sP(s) \dif s \right)\\ 
		&\le - \alpha(0) \left( P(t) - 4L P(t) \int_{0}^t s \dif s \right)\\ 
		&= - \alpha(0) P(t) \left( 1 - 2Lt^2 \right), 
	\end{align*}
	where the second inequality is due to the monotonicity of $P(s)$. 
	Since for all $0\le t\le 1/(2\sqrt{L})$ we have $P(t) \ge \mu t$ and $1-2Lt^2 \ge 1/2$, we deduce that
	\[
		\alpha'(t) \le - \alpha(0) \frac{\mu}{2} t.
	\]
	Finally, one more integration yields
	\[
		\alpha(t) = \alpha(0) + \int_{0}^t \alpha'(s) \dif s \le \alpha(0) \left( 1 - \frac{\mu}{4} t^2 \right), 
	\]
	and the theorem follows.
\end{proof}

\begin{proof}[Proof of Theorem~\ref{thm:rel_upper}]
	Lemma~\ref{lem:HMC_contraction} implies that for any constant $c\ge 2$, the Ricci curvature of ideal HMC with step-size $T=1/(c\sqrt{L})$ is at least $1/(8c^2\kappa)$. 
	Then, it follows from \cite[Proposition 29]{O09} that the spectral gap of ideal HMC is at least $1/(8c^2\kappa)$. Hence, the relaxation time is upper bounded by $8c^2\kappa = O(\kappa)$.
\end{proof}

\subsection{Proof of Claim~\ref{claim:crude-bound}}
\label{subsec:proof_claim}
We present the proof of Claim~\ref{claim:crude-bound} in this section. We remark that a similar crude bound was established in \cite{LSV18} for general matrix ODEs. Here we prove the crude bound specifically for the Hamiltonian ODE, but without assuming that $f$ is twice differentiable.
\begin{proof}[Proof of Claim~\ref{claim:crude-bound}]
	We first derive a crude upper bound on $\norm{u-v}$. Since $f$ is $L$-smooth, we have 
	\begin{align*}
		\frac{\dif}{\dif t} \norm{u-v} 
		&= \frac{- \ip{\grad f(x) - \grad f(y)}{u-v}}{\norm{u-v}}\\
		&\le \norm{\grad f(x) - \grad f(y)} 
		\le L \norm{x-y}.
	\end{align*}
	Then from $u_0 = v_0$ we get
	\[
		\norm{u-v} = \int_{0}^t \left( \frac{\dif}{\dif s} \norm{u-v} \right) \dif s \le L \int_{0}^t \norm{x-y} \dif s.
	\]
	
	To obtain the upper bound for $\norm{x-y}$, we first bound its derivative by
	\begin{equation}\label{eqn:abs_deri_x-y}
		\left| \frac{\dif}{\dif t} \norm{x-y} \right| = \frac{\left| \ip{u-v}{x-y} \right|}{\norm{x-y}} \le \norm{u-v} \le L \int_{0}^t \norm{x-y} \dif s.
	\end{equation}
	Therefore, 
	\begin{align*}
		\norm{x-y} &= \norm{x_0 - y_0} + \int_{0}^t \left( \frac{\dif}{\dif s} \norm{x-y} \right) \dif s\\ 
		&\le \norm{x_0 - y_0} + L \int_{0}^t \int_{0}^s \norm{x-y} \dif r \dif s\\
		&= \norm{x_0 - y_0} + L \int_{0}^t (t-s) \norm{x-y} \dif s.
	\end{align*}
	We then deduce from \cite[Lemma A.5]{LSV18} that
	\begin{equation}\label{eqn:upper_norm_x-y}
		\norm{x-y} \le \norm{x_0 - y_0} \cosh\left( \sqrt{L}t \right) \le \sqrt{2} \norm{x_0 - y_0},
	\end{equation}
	where we use the fact that $\cosh(\sqrt{L} t) \leq \cosh(1/2) \le \sqrt{2}$.
	
	Next, we deduce from \eqref{eqn:abs_deri_x-y} and \eqref{eqn:upper_norm_x-y} that
	\begin{align*}
		\frac{\dif}{\dif t} \norm{x-y} 
		&\ge - L \int_{0}^t \norm{x-y} \dif s\\ 
		&\ge - L \norm{x_0 - y_0} \int_{0}^t  \cosh\left( \sqrt{L}s \right) \dif s\\ 
		&= - \sqrt{L} \norm{x_0 - y_0} \sinh\left( \sqrt{L} s \right).
	\end{align*}
	Thus, we obtain
	\begin{align*}
		\norm{x-y} &= \norm{x_0 - y_0} + \int_{0}^t \left( \frac{\dif}{\dif s} \norm{x-y} \right) \dif s\\ 
		&\ge \norm{x_0 - y_0} - \sqrt{L} \norm{x_0 - y_0} \int_{0}^t \sinh\left( \sqrt{L} s \right) \dif s\\
		&= \norm{x_0 - y_0} \left( 2 - \cosh\left( \sqrt{L} t \right) \right) \ge \frac{1}{\sqrt{2}} \norm{x_0 - y_0},
	\end{align*}
	where we use $2 - \cosh(\sqrt{L} t) \ge 2- \cosh(1/2) \ge 1/\sqrt{2}$.
\end{proof}

\section{Lower bound for ideal HMC}
In this section, we show that the relaxation time of ideal HMC can achieve $\Theta(\kappa)$ for some $\mu$-strongly convex and $L$-smooth function, and thus prove Theorem~\ref{thm:rel_lower}. 

Consider a two-dimensional quadratic function:
\[
f(x_1,x_2) = \frac{x_1^2}{2 \sigma_1^2} + \frac{x_2^2}{2\sigma_2^2},
\]
where $\sigma_1 = 1/\sqrt{\mu}$ and $\sigma_2 = 1/\sqrt{L}$. Thus, $f$ is $\mu$-strongly convex and $L$-smooth. The probability density $\nu$ proportional to $e^{-f}$ is essentially the bivariate Gaussian distribution: for $(x_1,x_2)\in \R^2$, 
\[
\nu(x_1,x_2) = \frac{1}{2\pi \sigma_1 \sigma_2} \exp\left( - \frac{x_1^2}{2\sigma_1^2} - \frac{x_2^2}{2\sigma_2^2} \right).
\]
The following lemma shows that ideal HMC for the bivariate Gaussian distribution $\nu$ has relaxation time $\Omega(\kappa)$, and then Theorem~\ref{thm:rel_lower} follows immediately.

\begin{lemma}
	For any constant $c > 0$, the relaxation time of ideal HMC for sampling from $\nu$ with step-size $T = 1/(c\sqrt{L})$ is at least $2c^2\kappa$.
\end{lemma}

\begin{proof}
	The Hamiltonian curve for $f$ is given by the ODE
	\[
		(x_1'', x_2'') = - \grad f(x_1,x_2) = \left( -\frac{x_1}{\sigma_1^2}, -\frac{x_2}{\sigma_2^2} \right)
	\]
	with initial position $(x_1(0), x_2(0))$ and initial velocity $(x_1'(0),x_2'(0))$ from the bivariate standard Gaussian $N(0,I)$. Observe that $\nu = \nu_1 \otimes \nu_2$ is a product distribution of the two coordinates and HMC for $f$ is a product chain. Thus, we can consider the dynamics for each coordinate separately. The Hamiltonian ODE for one coordinate becomes
	\[
		x_i'' = - \frac{x_i}{\sigma_i^2}, \quad x_i(0), \quad x_i'(0) = v_i(0) \sim N(0,1)
	\]
	where $i=1,2$. Solving the ODE above and plugging in the step-size $t=T$, we get
	\[
		x_i(T) = x_i(0) \cos(T/\sigma_i) + v_i(0) \sigma_i \sin(T/\sigma_i).
	\]

	Let $P_i$ be the transition kernel of ideal HMC for the $i$th coordinate (considered as a Markov chain on $\R$). Then for $x,y\in\R$ we have
	\[
		P_i(x,y) = \frac{1}{\sqrt{2\pi} \sigma_i \sin(T/\sigma_i)} \exp\left( - \frac{\left( y - x\cos(T/\sigma_i) \right)^2}{2\sigma_i^2 \sin^2(T/\sigma_i)} \right).
	\] 
	Namely, given the current position $x$, the next position $y$ is from a normal distribution with mean $x\cos(T/\sigma_i)$ and variance $\sigma_i^2 \sin^2(T/\sigma_i)$. 
	Denote the spectral gap of $P_i$ by $\gamma_i$ for $i=1,2$ and that of ideal HMC by $\gamma$. Let $h(x) = x$ and note that $h\in L^0_2(\nu_i)$. Using the properties of product chains and spectral gaps, we deduce that
	\[
		\gamma \le \min\{\gamma_1,\gamma_2\} \le \gamma_1 
		= 1 - \sup_{f\in L_2^0(\nu_1)} \frac{|\ip{P_1 f}{f}|}{\norm{f}^2}
		\le 1 - \frac{|\ip{P_1 h}{h}|}{\norm{h}^2}. 
	\]
	Since we have
	\[
		\norm{h}^2 = \int_{-\infty}^{\infty} \nu_1(x) h(x)^2 \dif x = \sigma_1^2
	\]
	and
	\[
		\ip{P_1 h}{h} = \int_{-\infty}^{\infty} \nu_1(x) P_1(x,y) h(x) h(y) \dif x \dif y = \sigma_1^2 \cos(T/\sigma_1),
	\]
	it follows that $\gamma \le 1 - |\cos(T/\sigma_1)|$.
	Suppose that $T = 1/(c\sqrt{L})$ for some $c>0$. Then we get
	\[
		\gamma 
		\le \frac{T^2}{2\sigma_1^2} = \frac{1}{2c^2} \frac{\mu}{L}, 
	\]
	and consequently $\Trel = 1/\gamma \ge 2c^2\kappa$.
\end{proof}

\section{Convergence rate of discretized HMC}
In this section, we show how our improved contraction bound (Lemma~\ref{lem:HMC_contraction}) implies that HMC returns a good enough sample after $\widetilde{O}((\kappa d)^{1.5})$ steps. We will use the framework from \cite{LSV18} to establish Theorem~\ref{thm:mixing}.

We first state the ODE solver from \cite{LSV18}, which solves an ODE in nearly optimal time when the solution to the ODE can be approximated by a piece-wise polynomial. We state here only for the special case of second order ODEs for the Hamiltonian system. We refer to \cite{LSV18} for general $k$th order ODEs.

\begin{thm}[{\cite[Theorem 2.5]{LSV18}}]
	\label{thm:ODE_solver}
	Let $x(t)$ be the solution to the ODE
	\[
		x''(t) = -\grad f(x(t)), \quad x(0) = x_0, \quad x'(0) = v_0.
	\]
	where $x_0,v_0 \in \R^d$ and $0\le t\le T$. Suppose that the following conditions hold:
	\begin{enumerate}
		\item There exists a piece-wise polynomial $q(t)$ such that $q(t)$ is a polynomial of degree $D$ on each interval $[T_{j-1},T_j]$ where $0=T_0 < T_1 < \dots <T_m = T$, and for all $0\le t\le T$ we have
		\[
			\norm{q(t) - x''(t)} \le \frac{\delta}{T^2};
		\]
		\item $\{T_j\}_{j=1}^m$ and $D$ are given as input to the ODE solver;
		\item The function $f$ has a $L$-Lipschitz gradient; i.e., for all $x,y\in\R^d$,
		\[
			\norm{\grad f(x) - \grad f(y)} \le L \norm{x-y}.
		\]
	\end{enumerate}
	If $\sqrt{L} T \le 1/16000$, then the ODE solver can find a piece-wise polynomial $\tilde{x}(t)$ such that for all $0\le t\le T$,
	\[
		\norm{\tilde{x}(t) - x(t)} \le O(\delta).
	\]
	The ODE solver uses $O(m(D+1) \log(CT/\delta))$ evaluations of $\grad f$ and $O(dm (D+1)^2 \log(CT/\delta))$ time where
	\[
		C = O\left(\norm{v_0} + T\norm{\grad f(x_0)} \right).
	\]
\end{thm}

The following lemma, which combines Theorem~3.2, Lemma 4.1 and Lemma 4.2 from \cite{LSV18}, establishes the conditions of Theorem~\ref{thm:ODE_solver} in our setting. We remark that Lemmas 4.1 and 4.2 hold for all $T \le 1/(8\sqrt{L})$, and Theorem~3.2, though stated only for $T \le O(\mu^{1/4}/L^{3/4})$ in \cite{LSV18}, holds in fact for the whole region $T \le 1/(2\sqrt{L})$ where the contraction bound (Lemma~\ref{lem:HMC_contraction}) is true. We omit these proofs here and refer the readers to \cite{LSV18} for more details.

\begin{lemma}\label{lem:ODE_prep}
	Let $f$ be a twice differentiable function such that $\mu I \preceq \grad^2 f(x) \preceq LI$ for all $x\in \R^d$. Choose the starting point $x^{(0)} = \arg\min_x f(x)$, step-size $T = 1/(16000\sqrt{L})$, and ODE error tolerance $\delta = \sqrt{\mu}T^2 \eps/16$ in the HMC algorithm. Let $\{x^{(k)}\}_{k=1}^N$ be the sequence of points we get from the HMC algorithm and $\{v_0^{(k)}\}_{k=1}^N$ be the sequence of random Gaussian vector we choose in each step. Let $\pi \propto e^{-f}$ be the target distribution and let $\pi_{\textsc{hmc}}$ be the distribution of $x^{(N)}$, i.e., the output of HMC. For any $0<\eps<\sqrt{d}$, if we run HMC for 
	$$ N = O\left( \frac{\log(d/\eps)}{\mu T^2} \right) = O\left( \kappa \log(d/\eps) \right) $$ 
	steps where $\kappa = L/\mu$, then:
	\begin{enumerate}
		\item \emph{(\cite[Theorem 3.2]{LSV18})} We have that
		\[
			W_2(\pi_{\textsc{hmc}},\pi) \le \frac{\eps}{\sqrt{\mu}};
		\]
		\item \emph{(\cite[Lemma 4.1]{LSV18})} For each $k$, let $x_k(t)$ be the solution to the ODE \eqref{eq:ODE_Ham} in the $k$th step of HMC. Then there is a piece-wise constant function $q_k$ of $m_k$ pieces such that $\norm{q_k(t)-x_k''(t)} \le \delta/T^2$ for all $0\le t\le T$, where
		\[
		m_k = \frac{2LT^3}{\delta}\left( \norm{v_0^{(k-1)}} + T\norm{\grad f(x^{(k-1)})} \right);
		\]
		\item \emph{(\cite[Lemma 4.2]{LSV18})} We have that
		\[
		\frac{1}{N}\, \Exp\left[ \sum_{k=1}^{N} \norm{\grad f(x^{(k-1)})}^2 \right] \le O(Ld).
		\]
	\end{enumerate}		
\end{lemma}

\begin{proof}[Proof of Theorem~\ref{thm:mixing}]
	The convergence of HMC is guaranteed by part 1 of Lemma~\ref{lem:ODE_prep}. In the $k$th step, the number of evaluations of $\grad f$ is $O(m_k \log (C_k \sqrt{\kappa}/\eps))$ by Theorem~\ref{thm:ODE_solver} and part 2 of Lemma~\ref{lem:ODE_prep}, where
	\[
		m_k = O\left( \frac{\sqrt{\kappa}}{\eps}\right) \left( \norm{v_0^{(k-1)}} + T\norm{\grad f(x^{(k-1)})}  \right) 
	\mbox{ and }
		C_k = O\left( \norm{v_0^{(k-1)}} + T\norm{\grad f(x^{(k-1)})} \right).
	\]
	Thus, the average number of evaluations of $\grad f$ per step is upper bounded by
	\[
		\frac{1}{N}\, \Exp \left[ \sum_{k=1}^N O(m_k \log (C_k \sqrt{\kappa}/\eps)) \right] \le \frac{1}{N}\, \Exp \left[ \sum_{k=1}^N O(m_k \log m_k) \right]
		\le \frac{1}{N} \, O\left( \Exp \left[ M \log M \right] \right),
	\]
	where $M = \sum_{k=1}^N m_k$. Since each $v_0^{(k-1)}$ is sampled from the standard Gaussian distribution, we have $\Exp \Big[ \norm{v_0^{(k-1)}}^2 \Big] = d$. Thus, by the Cauchy-Schwarz inequality and part 3 of Lemma~\ref{lem:ODE_prep}, we get
	\[
		\Exp\left[ M^2 \right] \le N \sum_{k=1}^N \Exp\left[ m_k^2 \right] \le O\left( \frac{N\kappa}{\eps^2} \right) \sum_{k=1}^N \Exp\left[ \norm{v_0^{(k-1)}}^2 \right] + T^2 \Exp\left[ \norm{\grad f(x^{(k-1)})}^2 \right] \le O\left( \frac{N^2\kappa d}{\eps^2} \right).
	\]	
	We then deduce again from the Cauchy-Schwarz inequality that
	\[
		\left(\Exp[M \log M]\right)^2 \le \Exp\left[M^2\right] \cdot \Exp\left[\log^2 M\right] \le \Exp\left[ M^2 \right] \cdot \log^2\left( \Exp M \right) 
		\le \Exp[M^2] \cdot \log^2\left( \sqrt{\Exp\left[ M^2 \right]} \right), 
	\]
	where the second inequality is due to that $h(x) = \log^2(x)$ is concave when $x \geq 3$. Therefore, the number of evaluations of $\nabla f$ per step, amortized over all steps, is 
	\[
		\frac{1}{N}\, O\left( \sqrt{\Exp[M^2]}  \log\left( \sqrt{\Exp\left[ M^2 \right]} \right) \right) \le O\left( \frac{\sqrt{\kappa d}}{\eps} \log\left( \frac{\kappa d}{\eps} \right) \right).
	\]
	Using a similar argument we have the bound for the expected running time per step. This completes the proof.
\end{proof}

\paragraph{Acknowledgments.} We are grateful to Yin Tat Lee for helpful discussions. This work was supported in part by NSF awards CCF-1563838, 1717349, 1617306 and DMS-1839323.

 \bibliographystyle{plain}
\bibliography{HMC}

\end{document}